\documentclass{article}
\usepackage{arxiv}
\usepackage{natbib}        

\usepackage[utf8]{inputenc}
\usepackage{amsmath}
\usepackage{amsfonts}
\usepackage{amssymb}
\usepackage{amsthm}

\usepackage{algorithm}
\usepackage[all]{xy}

\usepackage{xcolor}  

\usepackage{amssymb}

\usepackage{amsmath}
\usepackage{mathdots}

\usepackage{mathtools}

\usepackage{tensor}

\usepackage{urwchancal}
\DeclareFontFamily{OT1}{pzc}{}
\DeclareFontShape{OT1}{pzc}{m}{it}{<-> s * [1.10] pzcmi7t}{}
\DeclareMathAlphabet{\mathpzc}{OT1}{pzc}{m}{it}

\usepackage{graphicx}
\usepackage{ifpdf}
\ifpdf
\usepackage{epstopdf}
\PrependGraphicsExtensions{.svg}
\fi

\newtheorem{theorem}{Theorem}[section]
\newtheorem{lemma}[theorem]{Lemma}

\newtheorem{proposition}[theorem]{Proposition}

\newtheorem{remark}[theorem]{Remark}

\providecommand{\R}{\mathbb{R}}

\providecommand{\SO}{\mathbf{SO}}

\providecommand{\GL}{\mathbf{GL}}

\providecommand{\grpG}{\mathbf{G}}

\providecommand{\gothg}{\mathfrak{g}}

\providecommand{\gothm}{\mathfrak{m}}

\providecommand{\gothX}{\mathfrak{X}} 

\providecommand{\so}{\mathfrak{so}}

\providecommand{\Sph}{\mathrm{S}}

\providecommand{\calG}{\mathcal{G}}

\providecommand{\calM}{\mathcal{M}}
\providecommand{\calN}{\mathcal{N}}

\providecommand{\calU}{\mathcal{U}}

\providecommand{\vecL}{\mathbb{L}}

\providecommand{\Sym}{\mathbb{S}} 

\providecommand{\tT}{\mathrm{T}} 

\providecommand{\eb}{\mathbf{e}} 

\DeclareMathOperator{\Ad}{Ad}

\providecommand{\id}{\mathrm{id}} 

\providecommand{\Lyap}{\mathcal{L}} 

\providecommand{\td}{\mathrm{d}}
\providecommand{\tD}{\mathrm{D}}

\providecommand{\ddt}{\frac{\td}{\td t}}

\providecommand{\mr}[1]{\mathring{#1}} 

\usepackage{accents}
\usepackage{mathtools}
\makeatletter
\providecommand{\scirc}{%
    \hbox{\fontfamily{\rmdefault}\fontsize{0.4\dimexpr(\f@size pt)}{0}\selectfont{\raisebox{-0.52ex}[0ex][-0.52ex]{$\circ$}}}}

\makeatother

\makeatletter
\providecommand{\ucirc}{%
    \hbox{\fontfamily{\rmdefault}\fontsize{0.4\dimexpr(\f@size pt)}{0}\selectfont{\raisebox{0.0ex}[0ex][-0.52ex]{$\circ$}}}}

\makeatother

\mathchardef\mhyphen="2D

\providecommand{\etal}{\textit{et al.~}}

\renewcommand{\mr}[1]{#1^\circ}

\usepackage{hyperref}

\providecommand{\tT}{\mathrm{T}}
\providecommand{\gothm}{\mathfrak{m}}
\providecommand{\lchart}{\vartheta}
\providecommand{\figwidth}{0.7}

\newcommand{\newhl}{}

\newcommand{\papercitation}{
P. van Goor, T. Hamel and R. Mahony, "Equivariant Filter (EqF)," in \textit{IEEE Transactions on Automatic Control}, vol. 68, no. 6, pp. 3501-3512, June 2023
}

\newcommand{\publicationdetails}
{
\papercitation \\
\copyrightNoticeIEEE{2022} 
}
\newcommand{\publicationversion}
{Author accepted version}

\begin{document}



\title{Equivariant Filter (EqF)}
\headertitle{Equivariant Filter (EqF)}


\author{
\href{https://orcid.org/0000-0001-2345-6789}{\includegraphics[scale=0.06]{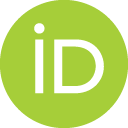}\hspace{1mm}
Pieter van Goor}
\\
    Systems Theory and Robotics Group \\
    Australian Centre for Robotic Vision \\
	Australian National University \\
    ACT, 2601, Australia \\
    \texttt{Pieter.vanGoor@anu.edu.au} \\
	\And	\href{https://orcid.org/0000-0001-2345-6789}{\includegraphics[scale=0.06]{orcid.png}\hspace{1mm}
    Tarek Hamel}
\\
    I3S (University C\^ote d'Azur, CNRS, Sophia Antipolis) \\
    and Insitut Universitaire de France \\
    \texttt{THamel@i3s.unice.fr} \\
	\And	\href{https://orcid.org/0000-0001-2345-6789}{\includegraphics[scale=0.06]{orcid.png}\hspace{1mm}
    Robert Mahony}
\\
    Systems Theory and Robotics Group \\
    Australian Centre for Robotic Vision \\
	Australian National University \\
    ACT, 2601, Australia \\
	\texttt{Robert.Mahony@anu.edu.au} \\
}

\maketitle

\begin{abstract}
The kinematics of many systems encountered in robotics, mechatronics, and avionics are naturally posed on homogeneous spaces; that is, their state lies in a smooth manifold equipped with a transitive Lie group symmetry.
This paper proposes a novel filter, the \emph{Equivariant Filter (EqF)}, by posing the observer state on the symmetry group, linearising global error dynamics derived from the equivariance of the system, and applying extended Kalman filter design principles.
We show that equivariance of the system output can be exploited to reduce linearisation error and improve filter performance.
Simulation experiments of an example application show that the EqF significantly outperforms the extended Kalman filter and that the reduced linearisation error leads to a clear improvement in performance.
\end{abstract}


\section{Introduction}


The importance of Lie group symmetries in analysing nonlinear control systems has been recognised since the 1970s \cite{1972_brocket_SIAM,1972_jurdjevic_JDE,1973_brocket_SIAM}.
Jurdjevic and Sussmann generalised the ideas of Brockett \cite{1972_brocket_SIAM} around the theory of systems on matrix Lie groups to abstract Lie groups \cite{1972_jurdjevic_JDE}.
Cheng \etal described necessary and sufficient conditions for observability of such systems
\cite{1990_cheng_SIAM}.
A comprehensive discussion of these early results can be found in Chapter 6 of Jurdjevic's book \cite{1997_jurdjevic_geometric_control}.

In one of the earliest works applying Lie group symmetry for observer design, Salcudean \cite{1991_salcudean_TAC} proposed a nonlinear observer for attitude estimation of a satellite using the quaternion representation of rotation.
Thienel \etal  \cite{2003_thienel_TAC} added an analysis of observability and  bias estimation.
Aghannan \cite{2003_aghannan_TAC} proposed a general observer design methodology for Lagrangian systems by exploiting invariance properties.
Driven by the emerging aerial robotics community and the need for robust and simple attitude observers, Mahony \etal \cite{2008_Mahony_tac} developed a nonlinear observer for rotation matrices posed directly on the matrix Lie group $\SO(3)$ with almost-globally asymptotically stable error dynamics.
In parallel, Bonnabel \etal \cite{2007_Bonnabel_cdc} proposed the left-invariant extended Kalman filter and applied this to attitude estimation.
Both of these observers are fundamentally derived from the symmetry properties of the underlying system and had significant impact in the robotics community.
This motivated more general studies of systems on Lie groups and homogeneous spaces (smooth manifolds with transitive Lie group actions) \cite{2008_Bonnabel_TAC,2009_Bonnabel_TAC,2009_lageman_TAC}.
In \cite{2009_Bonnabel_cdc}, Bonnabel \etal proposed the Invariant Extended Kalman Filter (IEKF): a sophisticated observer design for systems on Lie groups with invariance properties.
Given a system on a homogeneous space, a property of the system is termed \emph{equivariant} if it changes in a compatible way under transformation of the state by the symmetry Lie group.
Mahony \etal considered observer design for equivariant kinematic systems on homogeneous spaces with equivariant output functions using Lyapunov design principles \cite{RM_2013_Mahony_nolcos}.
This general design was extended in \cite{2015_khosravian_automatica} to consider biased input measurements and general gain mappings.
Work by Barrau and Bonnabel \cite{2017_Barrau_tac} extended the IEKF from invariant systems to a broader class of ``group affine'' systems, again focusing on Lie groups only, and characterised the filter's convergence properties.
Lie group variational integrators, first proposed in \cite{2004_leok_foundations}, were applied to discretise an observer for rigid body attitude estimation by Itazi and Sanyal in \cite{2014_izadi_automatica}.
In \cite{2016_johansen_ecc}, Johansen and Fossen proposed the exogenous Kalman filter, which Stovner \etal \cite{2018_Stovner_automatica} showed to be globally exponentially stable when applied to attitude estimation by taking advantage of the invariant Lie group dynamics.
Roumeliotis \etal \cite{1999_roumeliotis_icra} applied the error-state Kalman filter \cite[Chapter 6.3]{1982_maybeck_stochastic} to attitude estimation for mobile robot localisation using the group product between unit quaternions to define the error state.
Sola \cite{2017_sola_quaternion} recently provided a detailed description of the error-state Kalman filter for inertial-aided navigation, including the use of the unit quaternion group product for the definition of attitude error.
In the same spirit, Hamel and Samson \cite{2017_hamel_tac}  constructed a general Riccati observer for a class of systems of time-varying nonlinear systems and showed the local exponential stability of the origin of observer error  as long as  the linearised system about the true state is uniformly observable.
Barrau \etal \cite{2016_Barrau_arxive} provided a novel Lie group that models the classical Simultaneous Localisation and Mapping (SLAM) problem in robotics.
Parallel work \cite{2017_Mahony_cdc} found the same structure and showed how it also models invariance in the SLAM problem, leading to a homogeneous state space structure.
This led to recent work exploiting equivariance in visual (inertial) odometry \cite{2021_vangoor_auto,2020_mahony_cdc,2021_goor_EquivariantFilterVisual} where there is no direct Lie group structure for the state space, and filter design methods for which a Lie group structure is necessary cannot be applied.
The present paper draws from recent work on observer design specifically targeting systems on homogeneous spaces \cite{2020_vangoor_cdc,2021_mahony_EquivariantFilterDesign, 2021_mahony_arcar}.

In this paper, we propose the Equivariant Filter (EqF): a novel filter design for equivariant kinematic systems posed on homogeneous spaces.
The filter is derived by exploiting the Lie group symmetry of the kinematic system to derive global error coordinates.
The EqF observer dynamics are defined on the symmetry Lie group; however, the correction is computed using a Riccati equation associated with linearised error kinematics about a fixed origin on the homogeneous state space.
The proposed architecture fully exploits the symmetry properties of the system without requiring that the system model is posed explicitly on the Lie group and applies to any system with just the basic equivariance property.
In contrast, the existing state-of-the-art observer/filter design methodologies for systems with symmetry depend on assuming properties of the system: invariance for the constructive designs \cite{RM_2013_Mahony_nolcos} or the more general group affine structure that is only defined for systems with Lie group state space for the IEKF designs \cite{2017_Barrau_tac}.
Interestingly, the proposed EqF specialises to the IEKF \cite{2017_Barrau_tac} when the system considered is posed directly on a Lie group and displays the specific group affine structure required for the IEKF derivation.
Barrau \etal \cite{2017_Barrau_tac} showed that the IEKF admits an exact linearisation of the deterministic part of the state equation of the error dynamics, a property that the EqF shares on compatible systems, leading to significant performance gains versus an EKF derived without regard for the symmetry.
In addition, the EqF is designed to accommodate symmetries that are compatible with the configuration output.
When such a symmetry is used, we propose a novel approximation of the output equation that eliminates second-order error in the output linearisation.
Implementing the equivariant filter design methodology with this approximation leads to improved performance and we term the resulting observer the EqF$^\star$. 
The EqF methodology extends existing filter design methodologies for invariant systems by both applying to a broad class of systems on homogeneous spaces (rather than only systems with Lie group state space) and exploiting equivariance of the configuration output to further enhance filter performance compared to state-of-the-art.

We demonstrate the potential of the EqF and EqF$^\star$ using a simple example of single-bearing estimation; that is, the problem of determining the bearing of a fixed direction (with state space the sphere $\Sph^2$) in an inertial frame with respect to a rotating frame with known angular velocity.
There are no global coordinates on the sphere and application of the EKF requires consideration of local or embedded coordinates.
Moreover, the state space is not a Lie group, precluding the direct application of filter design methods that require a Lie group state space.
We simulate a classical EKF, the EqF and the EqF$^\star$ on this system.
The simulation results demonstrate the known advantage \cite{2017_Barrau_tac} of exploiting symmetry versus (even a careful) EKF design in local coordinates and goes on to demonstrate a significant performance advantage for the EqF$^\star$ over the standard EqF.
The interested reader can also find a more tutorial exposition of equivariant observer design with a discussion of the equivariant filter in the preprint \cite{2021_mahony_arcar}.


In \S\ref{sec:motivating_example}, the design of an EKF for the example of single-bearing estimation is detailed to motivate the developments in the sequel.
In \S\ref{sec:prelim}, we define key notation and provide preliminary results.
General systems on homogeneous spaces are defined and discussed in \S\ref{sec:problem_description}.
In \S\ref{sec:eqf_design}, the notion of a lifted system is used to develop the dynamics of a global state error by exploiting the Lie symmetry.
We also discuss the linearisation of the error dynamics and show that the existence of an equivariant output leads to a better (lower error) linearisation of the output map.
In \S\ref{sec:eqf_definition}, the EqF equations are presented and we provide some insight into tuning the filter in practice.
The example problem of single-bearing estimation is revisited in \S\ref{sec:sphere_example}.
We show simulation results to demonstrate the performance of the EqF and EqF$^\star$ compared to an EKF.
In \S\ref{sec:conclusion}, we provide concluding remarks.
In Appendix \ref{sec:filter_implementation}, we provide a step-by-step design methodology for implementing an EqF, and in Appendix \ref{sec:iekf_specialisation} we show how the EqF specialises to the invariant extended Kalman filter for a specific subclass of systems on Lie groups.
We also provide open source code\footnote{\url{https://github.com/STR-ANU/auto_eqf}} for the implementation of the EqF for general systems based on numerical differentiation.

\section{Motivating Example: Single Bearing Estimation}
\label{sec:motivating_example}

In this section we present the problem of single bearing estimation.
This example has been chosen to be as simple as possible algebraically while presenting a problem where the equivariant filter approach is of interest.
Consider a robot equipped with a gyroscope that measures its angular velocity $\Omega \in \R^3$ and a magnetometer that measures the magnetic field in the robot's body-fixed frame $\eta \in \Sph^2$.
The noise free dynamics of $\eta$ are
\begin{align}
    \dot{\eta} &= f_{\Omega}(\eta) := -\Omega^\times \eta. \label{eq:ex_sphere_dynamics}
\end{align}
where $\Omega^\times \in \R^{3 \times 3}$ is the matrix
\begin{align*}
    \Omega^\times :=
    \begin{pmatrix}
        0 & -\Omega_3 & \Omega_2 \\
        \Omega_3 & 0 & -\Omega_1 \\
        -\Omega_2 & \Omega_1 & 0
    \end{pmatrix}.
\end{align*}
The noise free measurement model considered is a 3-axis magnetometer
\begin{align}
    y = h(\eta) = c_m \eta \in \R^3 \label{eq:ex_sphere_output},
\end{align}
where $c_m$ is the (known) magnetic field strength.
The system is nonlinear due to the manifold structure of the sphere $\Sph^2$ comprising both its state and its measurement.

Since the system is nonlinear, a nonlinear filter design is required and the extended Kalman filter is the industry standard choice.
To provide context for the main contribution of the paper we provide a sketch of the derivation of a classical EKF on the sphere for direction estimation.
In Section \ref{sec:sphere_example} we will provide simulation results that compare the classical EKF for this problem to the proposed EqF that is the main contribution of the paper.

\textbf{EKF for direction estimation:}
\newhl{
The EKF requires a system to be written in Euclidean coordinates and one must either choose local coordinates for the sphere or embed the sphere in $\R^3$ and introduce a state constraint in the filter design \cite{1990_tahk_tac}.
The sphere $\Sph^2$ has many possible choices of local coordinates, some of which cover the space almost globally.
However, all choices introduce significant nonlinearities in the system dynamics leading to loss of filter performance.
To avoid this and provide a state-of-the-art EKF implementation, we instead choose to implement the EKF by embedding the system dynamics \eqref{eq:ex_sphere_dynamics} and measurement \eqref{eq:ex_sphere_output} from the sphere into $\R^3$ and incorporating a nonlinear state constraint as in \cite{1990_tahk_tac}.
That is, the extended noise free system in $\R^3$ is
\begin{align}
    \dot{\eta} &= -\Omega^\times \eta, \notag \\
    h(\eta) &= c_m \frac{\eta}{\vert \eta \vert}, \notag \\
    g(\eta) & = \| \eta \|^2 = 1, \label{eq:ex_sphere_constraint}
\end{align}
where $g$ is the nonlinear state constraint, and can be treated as an additional measurement with a small associated uncertainty due to linearisation error \cite{1990_tahk_tac}.
This approach is particularly attractive for direction estimation on the sphere since the embedded dynamics \eqref{eq:ex_sphere_dynamics} on $\R^3$ are linear time-varying.
}
%
Let $\hat{\eta} \in \R^3$  be the EKF state estimate and define the state error and innovation $\tilde{\eta} := \eta - \hat{\eta} \in \R^3$ and $\tilde{y} := h(\eta) - h(\hat{\eta})$, respectively.
The output linearisation is given by
 \begin{align}
     \tilde{y}
     &= \frac{c_m}{\vert \hat{\eta} \vert} \left(
     I_3 - \frac{\hat{\eta}\hat{\eta}^\top}{\hat{\eta}^\top\hat{\eta}}
     \right) \tilde{\eta} + O(\vert \tilde{\eta} \vert^2),
     \label{eq:EKF_C1_t}
 \end{align}
noting that $\hat{\eta} \in \R^3$.
The constraint innovation $\tilde{z} = g(\eta) - g(\hat{\eta})$ has linearisation
\begin{align}
     \tilde{z} & = - 2 \hat{\eta}^\top \tilde{\eta} + O(\vert \tilde{\eta} \vert^2).
 \end{align}
The linearised constraint is assigned a nonzero virtual covariance to reflect the linearisation error, and this choice is a design parameter without stochastic justification.

The proposed EKF described above delivers the best performance for a classical EKF design for the direction estimation problem that the authors are aware of.
However, it is clear that applying a classical EKF to this problem is not straightforward.
The practitioner must choose between local and embedded coordinates.
The former are not globally well-defined and introduce linearisation errors.
The latter require the introduction and linearisation of state constraints.
The resulting algorithms are not intrinsic in any sense, depending on a whole sequence of design decisions, and although the performance of the filters is accepted, a question always remains about whether a different choice of coordinates, or an embedding with a different state constraint could have improved the results.

\section{Preliminaries}
\label{sec:prelim}

For a comprehensive introduction to smooth manifolds and Lie groups, the authors recommend \cite{2012_lee_SmoothManifolds}.

For a smooth manifold $\calM$, let $\tT_\xi \calM$ denote the tangent space of $\calM$ at $\xi$, let $\tT\calM$ denote the tangent bundle, and let $\gothX(\calM)$ denote the infinite dimensional linear vector space of vector fields over $\calM$.
Given a vector field $f \in \gothX(\calM)$, $f (\xi) \in \tT_\xi \calM$ denotes the value of $f$ at $\xi \in \calM$.

Given a differentiable function between smooth manifolds $h:\calM \to \calN$, the linear map
\begin{align*}
    \tD_\xi |_{\xi'} h(\xi): \tT_{\xi'} \calM &\to  \tT_{h({\xi'})} \calN, \\
    v &\mapsto \tD_\xi |_{\xi'} h(\xi)[v],
\end{align*}
denotes the differential of $h$ with respect to the argument $\xi$ evaluated at $\xi'$.
The shorthand $\tD h$ is used when the argument and base point are implied.
In general, the composition of two maps $h_1$ and $h_2$ is written $h_1 \circ h_2$, with $h_1 \circ h_2(\xi) := h_1(h_2(\xi))$.
For linear maps $H_1, H_2$ the composition may also be written $H_1 \cdot H_2$ or $H_1 H_2$ to emphasise the linearity and the link to matrix multiplication.
We apply this notation frequently in the context of the chain rule,
\begin{align*}
    \tD_\xi |_{\xi'} (h_1 \circ h_2)(\xi) [v]
    &= \tD_\eta |_{h_2(\xi')} h_1(\eta)
    \cdot \tD_\xi |_{\xi'} h_2(\xi) [v], \\
    \tD (h_1 \circ h_2)[v]
    &= \tD h_1  \tD h_2 [v].
\end{align*}

A general Lie group is denoted $\grpG$ and has Lie algebra $\gothg$.
The identity element is written $\id \in \grpG$.
For any $X \in \grpG$, the left and right translations by $X$ are denoted $L_X$ and $R_X$, respectively, and are defined by
\begin{align*}
    L_X(Y) := XY, \qquad R_X(Y) := YX,
\end{align*}
where $XY$ denotes the group product between $X$ and $Y$.
The adjoint map $\Ad: \grpG \times \gothg \to \gothg$ is defined by
\begin{align*}
    \Ad_X [U] = \tD L_X \cdot \tD R_{X^{-1}} [U],
\end{align*}
for every $X \in \grpG$ and $U \in \gothg$.

A right action of a Lie group $\grpG$ on a manifold $\calM$ is a smooth map $\phi: \grpG \times \calM \to \calM$ that satisfies
\begin{align*}
    \phi(Y, \phi(X, \xi)) & = \phi(XY, \xi), \\
    \phi(\id, \xi) & = \xi,
\end{align*}
for any $X,Y \in \grpG$ and any $\xi \in \calM$.
We will choose all symmetries in this paper to be right actions, noting that any left action can be transformed to a right action by considering the inverse parameterization of the group \cite{2021_mahony_EquivariantFilterDesign}.
Right-handed symmetries are naturally associated with the body-fixed sensor suites typical for most mobile robot applications.
For a fixed $X \in \grpG$, the partial map $\phi_X : \calM \to \calM$ is defined by
$\phi_X(\xi) := \phi(X,\xi)$.
Likewise, for a fixed $\xi \in \calM$, the partial map $\phi_\xi : \grpG \to \calM$ is defined by
$\phi_\xi(X) := \phi(X,\xi)$.
An action $\phi$ is called \emph{transitive} if, for any $\xi, \xi' \in \calM$, there exists $X \in \grpG$ such that $\phi(X, \xi) = \xi'$.

A homogeneous space $\calM$ is a manifold with a smooth and transitive symmetry action $\phi : \grpG \times \calM \to \calM$, where $\grpG$ is a Lie group.
Let $m$ denote the dimension of $\calM$.
For any element $\mr{\xi} \in \calM$ one may always choose an $m$-dimensional subspace $\gothm \subset \gothg$ such that $\tD_{E | \id} \phi(E, \mr{\xi})$ is a linear isomorphism $\gothm \to \tT_{\mr{\xi}} \calM$.
Let $\cdot^\wedge : \R^m \to \gothm \subset \gothg$ be a linear isomorphism that identifies $\gothm$ with $\R^m$, and let its inverse be $\cdot^\vee : \gothm \to \R^m$.
Then, at least in a local neighbourhood $\calU_{\mr{\xi}}$ of $\mr{\xi}$, the map $\lchart : \calU_{\mr{\xi}} \subset \calM \to \R^m$, defined by the unique element $\lchart(\xi) \in \R^m$ such that $\phi(\exp(\lchart(\xi)^\wedge), \mr{\xi}) = \xi$, is well defined and smooth.
The map $\lchart : \calU_{\mr{\xi}} \to \R^m$ is called a normal coordinate chart of the homogeneous space (about $\mr{\xi}$) \cite{1963_kobayashi_foundations}.
Note that normal coordinates are not unique since the choice of subspace $\gothm$ and its identification with $\R^m$ are arbitrary.

\begin{proposition}
Any right action $\phi: \grpG \times \calM \to \calM$ induces a right action on the vector fields over $\calM$, denoted $\Phi: \grpG \times \gothX(\calM) \to \gothX(\calM)$, and defined by
\begin{align}
    \Phi (X, f) := \tD \phi_X \cdot f \circ \phi_X^{-1},
\label{eq:Phi}
\end{align}
for any $f \in \gothX(\calM)$ and $X \in \grpG$.
For a fixed $X \in \grpG$, $\Phi_X$ is a linear map on $\gothX(\calM)$.
\end{proposition}

\begin{proof}
For fixed $X \in \grpG$, $\phi_X : \calM \to \calM$ is a diffeomorphism and the map $\Phi(X, \cdot)$ is the push forward operator \cite{2012_lee_SmoothManifolds}.
To see that $\Phi$ is a group action, let $X,Y \in \grpG$ and $f \in \gothX(\calM)$.
Then,
\begin{align*}
    \Phi(Y, \Phi(X, f))
    &= \tD \phi_Y \tD \phi_X f \circ \phi_X^{-1} \circ \phi_Y^{-1}, \\
    &= \tD \phi_{XY} f \circ \phi_{Y^{-1}X^{-1}}, \\
    &= \Phi(XY, f).
\end{align*}
To prove linearity
let $c_1,c_2 \in \R$ and let $f_1, f_2 \in \gothX(\calM)$, then
\begin{align*}
    & \Phi(X, c_1 f_1 + c_2 f_2)(\xi) \\
    &= \tD \phi_X (c_1 f_1 + c_2 f_2)({\phi_X^{-1}(\xi)}), \\
    &= c_1 \tD \phi_X (f_1)({\phi_X^{-1}(\xi)}) + c_2 \tD \phi_X (f_2)({\phi_X^{-1}(\xi)}), \\
    &= c_1 \Phi(X, f_1)(\xi) + c_2 \Phi(X, f_2)(\xi),
\end{align*}
where the second-last line follows from the linearity of the differential $\tD \phi_X$.
\end{proof}

\section{Problem Description}
\label{sec:problem_description}

\subsection{Systems on Homogeneous Spaces}

Let $\calM$ be a smooth $m$-dimensional manifold termed the \textit{state space}.
An affine (kinematic) system on $\calM$ may be written
\begin{align*}
    \dot{\xi} = f_0(\xi) + \sum_i f_i (\xi) u_i,
\end{align*}
for some vector fields $f_0, f_1,..., f_l \in \gothX(\calM)$ and scalar input signals $u_1, ..., u_l \in \R$.
Such a system is represented by an affine system function \cite{2021_mahony_arcar}
\begin{align} \label{eq:kinematic_system_function}
    f: \vecL &\to \gothX(\calM), \notag \\
    u &\mapsto f_{u} \in \gothX(\calM),
\end{align}
where $f_u(\xi) := f_0(\xi) + \sum_i f_i (\xi) u_i$ and $\vecL$ is a real vector space termed the \textit{input space} with $u = (u_1,...,u_l) \in \vecL$ the components of $u$.
We refer to $f_0$ as the drift term of the system and $f_i$ as the input vector fields.
Trajectories $\xi(t) \in \calM$ on a time interval $[0,\infty)$ of the system considered are solutions of the ordinary differential equation
\begin{align}
    \dot{\xi} &= f_{u(t)}(\xi), & \xi(0) &\in \calM, \label{eq:kinematic_system}
\end{align}
with initial condition $\xi(0)$ and measured input signal $u(t) \in \vecL$.
We will assume $u(t)$ is sufficiently smooth to ensure unique well-defined solutions for all time.
The configuration output \cite{2021_mahony_arcar} for a kinematic system is a function
\begin{align} \label{eq:measurement_system}
    h:\calM \to \calN \subset \R^n,
\end{align}
where $\calN$ is a smooth manifold termed the \emph{output space}, embedded in $\R^n$.

Let $\grpG$ be a Lie group with Lie algebra $\gothg$, and suppose that $\calM$ is a homogeneous space of $\grpG$; that is, there exists a smooth, transitive, right group action of $\grpG$ on $\calM$,
\begin{align} \label{eq:phi_action}
    \phi: \grpG \times \calM \to \calM.
\end{align}


A \textit{lift} \cite{2021_mahony_arcar} for the system function \eqref{eq:kinematic_system} is a map $\Lambda: \calM \times \vecL \to \gothg$ satisfying
\begin{align} \label{eq:lift_condition}
    \tD_X |_\id \phi_\xi(X)\left[ \Lambda(\xi, u) \right] &= f_{u}(\xi),
\end{align}
for every $\xi \in \calM$ and $u \in \vecL$.
Any kinematic system \eqref{eq:kinematic_system} defined on a homogeneous space admits a lift $\Lambda: \calM \times \vecL \to \gothg$ satisfying \eqref{eq:lift_condition} \cite{RM_2013_Mahony_nolcos, 2021_mahony_arcar}.
In the particular case where $\tD_X |_\id \phi_\xi(X)$ is invertible ($\phi$ is a free group action \cite{2012_lee_SmoothManifolds}), the lift is unique \cite{2021_mahony_arcar}.

\subsection{Equivariant Systems}
\label{sec:equivariant_systems}
Equivariance of a system is a powerful structural property that can be formulated for any system on a homogeneous space.
There are many established examples of equivariant systems  \cite{2008_Mahony_tac,2008_Bonnabel_TAC,2013_zamani_tac,2014_izadi_automatica,2015_khosravian_automatica,2017_Barrau_tac,2018_Stovner_automatica,2021_vangoor_auto} where exploiting the equivariant structure has led to high performance observers and filters.

A kinematic system \eqref{eq:kinematic_system} is termed \textit{equivariant} if there exists a smooth right group action $\psi: \grpG \times \vecL \to \vecL$ such that
\begin{align*}
    \tD \phi_X f_{u}(\xi) &= f_{\psi_X(u)}(\phi_X(\xi)),
\end{align*}
for all $u \in \vecL$, $\xi \in \calM$ and $X \in \grpG$.
In this case, recalling \eqref{eq:Phi}, one has
\begin{align}
    f_{\psi_X(u)}(\xi)
    &= f_{\psi_X(u)}(\phi_X(\phi_X^{-1}(\xi))), \notag \\
    &= \tD \phi_X f_{u}(\phi_X^{-1}(\xi)), \notag \\
    &= \Phi_X f_{u}(\xi). \label{eq:equivalent_def_equivariance}
\end{align}
That is, $f_{\psi_X(u)} = \Phi_X f_{u}$ as vector fields on $\calM$, or equivalently, the diagram
\begin{align*}
    \xymatrix{
        \vecL \ar[r]_{\psi_X} \ar[d]_f & \vecL \ar[d]_f \\
        \gothX(\calM) \ar[r]_{\Phi_X} & \gothX(\calM)
    }
\end{align*}
commutes for every $X \in \grpG$.

Suppose the system \eqref{eq:kinematic_system} is equivariant with state symmetry $\phi : \grpG \times \calM \to \calM$ and input symmetry $\psi : \grpG \times \vecL \to \vecL$.
A lift $\Lambda$ for the system is equivariant if
\begin{align} \label{eq:equivariant_lift}
    \Lambda(\phi(X, \xi), \psi(X, u)) = \Ad_{X^{-1}} \Lambda(\xi, u),
\end{align}
for all $X \in \grpG$, $\xi \in \calM$, and $u \in \vecL$, or equivalently, the diagram
\begin{align*}
    \xymatrixcolsep{4pc}
    \xymatrix{
        \gothg \ar[r]_{\Ad_{X^{-1}}} & \gothg \\
        \calM \times \vecL \ar[r]_{\phi_X \times \psi_X} \ar[u]_\Lambda & \calM \times \vecL \ar[u]_\Lambda
    }
\end{align*}
commutes for every $X \in \grpG$.

\begin{remark} \label{rmk:lift}
On a matrix Lie group where the group action is right translation, the lift is just $\Lambda(X,u) = X^{-1} \dot{X} = X^{-1} f_u(X) \in \gothg$.
However, for a general system on a homogeneous space there may be many choices of equivariant lift.
In such a case, an equivariant lift can be found by expanding the conditions \eqref{eq:lift_condition} and \eqref{eq:equivariant_lift} for the particular system, and searching for a solution for $\Lambda$ \cite{2021_mahony_EquivariantFilterDesign}.
\end{remark}

The system \eqref{eq:kinematic_system_function} has \emph{equivariant output} if there exists an action $\rho : \grpG \times \calN \to \calN$ satisfying
\begin{align}
    \rho(X, h(\xi)) = h(\phi(X, \xi)),
\end{align}
for all $X \in \grpG$ and $\xi$, or equivalently, the diagram
\begin{align*}
\xymatrix{
    \calM \ar[r]_{\phi_X} \ar[d]_h & \calM \ar[d]_h \\
    \calN \ar[r]_{\rho_X} & \calN
}
\end{align*}
commutes for every $X \in \grpG$.

\section{Equivariant Error System}
\label{sec:eqf_design}


The filter design proposed in the sequel can be applied to any equivariant system where an equivariant lift can be found.
In \cite{2021_mahony_EquivariantFilterDesign} the authors showed that any system on a homogeneous space can be extended to an equivariant system, and for any equivariant system an equivariant lift can always be constructed although the resulting construction may be infinite dimensional in the input space.
Thus, in principle, the proposed EqF design applies to all systems on homogeneous spaces, although the authors note that the primary systems of interest are those which are equivariant directly, or for which the extension terminates in a finite dimensional input space.

\subsection{Observer Architecture}

Consider an equivariant kinematic system \eqref{eq:kinematic_system} with state symmetry $\phi: \grpG \times \calM \to \calM$ and input symmetry $\psi: \grpG \times \vecL \to \vecL$.
Let $\Lambda : \calM \times \vecL \to \gothg$ be an equivariant lift for this system.
Given a fixed but arbitrary $\mr{\xi} \in \calM$ and a known input signal $u(t) \in \vecL$, the \textit{lifted system} \cite{2021_mahony_EquivariantFilterDesign} is defined by the ODE
\begin{align} \label{eq:lifted_system}
    \dot{X} &= \tD L_X \Lambda(\phi(X, \mr{\xi}), u),
    \quad    \phi(X(0), \mr{\xi}) = \xi(0),
\end{align}
where $X(t) \in \grpG$.
The trajectory of the lifted system projects down to the original system trajectory \cite{2021_mahony_EquivariantFilterDesign} by
\begin{align} \label{eq:lifted_system_projection}
    \phi(X(t), \mr{\xi}) \equiv \xi(t).
\end{align}
Define the observer state to be an element of the group $\hat{X} \in \grpG$, and use the lifted system as the internal model for the observer dynamics
\begin{align} \label{eq:observer_ode}
    \dot{\hat{X}} &= \tD L_{\hat{X}} \Lambda(\phi(\hat{X}, \mr{\xi}), u) + \tD R_{\hat{X}} \Delta, \quad  \hat{X}(0) = \id,
\end{align}
where the correction term $\Delta$ remains to be chosen \cite{2021_mahony_EquivariantFilterDesign,2021_mahony_arcar}.

The state estimate of the observer is given by the projection
\[
\hat{\xi} =  \phi(\hat{X}(t), \mr{\xi}).
\]
Thus, the state of the observer is posed on the symmetry group rather than the state space of the system kinematics.
The correction term $\Delta$ will be chosen by applying a Riccati observer to global error dynamics linearised about the fixed origin $\mr{\xi}$.

\subsection{Global Error Dynamics}

Let $\xi \in \calM$ be the true state of the system.
Choose an arbitrary fixed origin $\mr{\xi} \in \calM$ and let $\hat{X} \in \grpG$ be a state observer with dynamics given by \eqref{eq:observer_ode}.
Define the \emph{global state error}
\begin{align} \label{eq:state_error_dfn}
    e := \phi(\hat{X}^{-1},\xi).
\end{align}
Note that  $\phi(\hat{X},\mr{\xi}) = \xi$ if and only if $e = \mr{\xi}$.
Therefore, the goal of the filter design will be to drive $e \to \mr{\xi}$.
Define the \emph{origin velocity}
\begin{align} \label{eq:u_prime}
\mr{u} := \psi(\hat{X}^{-1}, u).
\end{align}
Note that the origin velocity $\mr{u}(t)$ can always be constructed since both $\hat{X}(t)$ and $u(t)$ are available to the observer.
The action $\psi(\hat{X}^{-1}, u)$ is the equivariant system generalisation of the well known adjoint action $\Ad_{\hat{X}} U$ that transforms between left and right invariant algebra elements for invariant vector fields on a Lie group.
In the error dynamics \eqref{eq:state_error_kinematics} derived below, the action $\psi_{\hat{X}^{-1}}$ transforms the measured system input into the correct representation for the error dynamics around the chosen origin $\mr{\xi}$.

\begin{lemma}
Let the global state error $e$ be defined as in \eqref{eq:state_error_dfn} and the origin velocity $\mr{u}$ be defined as in \eqref{eq:u_prime}.
The dynamics of $e$ are given by
\begin{align}
    \dot{e} = \tD \phi_{e} \left( \Lambda(e, \mr{u}) - \Lambda(\mr{\xi}, \mr{u}) - \Delta \right), \label{eq:state_error_kinematics}
\end{align}
and depend only on $e$, $\mr{u}$ and the correction $\Delta$.
\end{lemma}

\begin{proof}
Let $X(t)$ be a solution to the lifted system \eqref{eq:lifted_system} satisfying $\phi(X(0), \mr{\xi}) = \xi(0)$.
Then $\xi \equiv \phi(X, \mr{\xi})$, and
\begin{align} \label{eq:error_state_group_rel}
    e = \phi(\hat{X}^{-1},\xi)
    = \phi(X \hat{X}^{-1}, \mr{\xi})
    = \phi(E, \mr{\xi}),
\end{align}
where $E := X \hat{X}^{-1}$.
Computing the dynamics of $E$, one has
\begin{align} \label{eq:group_error_kinematics}
    \dot{E} &= \tD R_{\hat{X}^{-1}} \tD L_X \Lambda(\phi(X, \mr{\xi}), u) \notag \\
    &\hspace{0.5cm} - \tD L_X \tD L_{\hat{X}^{-1}} \tD R_{\hat{X}^{-1}} \left( \tD L_{\hat{X}} \Lambda(\phi(\hat{X}, \mr{\xi}), u) + \tD R_{\hat{X}} \Delta \right), \notag \\
    &= \tD L_E \Ad_{\hat{X}} \left( \Lambda(\phi(X, \mr{\xi}), u) - \Lambda(\phi(\hat{X}, \mr{\xi}), u) \right) - \tD L_E \Delta, \notag \\
    &= \tD L_E \left( \Lambda(\phi(E, \mr{\xi}), \psi(\hat{X}^{-1}, u))
    - \Lambda(\mr{\xi}, \psi(\hat{X}^{-1}, u)) \right)
    - \tD L_E \Delta,
\end{align}
where the last line follows from the equivariance of the lift.
Recalling \eqref{eq:u_prime}, the dynamics of $e$ follow from \eqref{eq:error_state_group_rel} and \eqref{eq:group_error_kinematics},
\begin{align}
\dot{e}
&= \tD \phi_{\mr{\xi}} \dot{E}, \notag \\
&= \tD \phi_{\phi(E, \mr{\xi})} \tD L_{E^{-1}} \dot{E}, \notag \\
&= \tD \phi_{e} \left( \Lambda(\phi(E, \mr{\xi}), \mr{u}) - \Lambda(\mr{\xi}, \mr{u}) - \Delta \right), \notag \\
&= \tD \phi_{e} \left( \Lambda(e, \mr{u}) - \Lambda(\mr{\xi}, \mr{u}) - \Delta \right), \notag
\end{align}
as required.
\end{proof}


\subsection{Linearisation}


\subsubsection{Error Dynamics}

Let $e \in \calM$ denote the global state error \eqref{eq:state_error_dfn}, and let $\mr{u} \in \vecL$ denote the origin velocity \eqref{eq:u_prime}.
Fix a local coordinate chart $\lchart : \calU_{\mr{\xi}} \to \R^m$ where $\calU_{\mr{\xi}} \subset \calM$ is a neighbourhood of the fixed origin $\mr{\xi}$, and $\lchart(\mr{\xi}) = 0$.
Let $\varepsilon$ be the local coordinates of the state error,
\begin{align} \label{eq:linearised_state_error_dfn}
    \varepsilon = \lchart(e).
\end{align}
The EqF correction is designed by linearising the pre-observer error dynamics of $\varepsilon$ at zero; that is, the dynamics \eqref{eq:state_error_kinematics} with the correction term $\Delta \equiv 0$ set to zero.

\begin{lemma} \label{prop:epsilon_dynamics}
    Let $\lchart$ be local coordinates on $\calM$ in an open neighbourhood $\calU_{\mr{\xi}} \subset \calM$ around $\mr{\xi}$.
    Assume $e(t) \in \calU_{\mr{\xi}}$ for all time.
    The linearised pre-observer ($\Delta \equiv 0$) dynamics of $e(t)$ about $\varepsilon = 0$ are
    \begin{align}
        \dot{\varepsilon} &= \mr{A}_t \varepsilon
        + O(\vert \varepsilon \vert^2),
        \label{eq:linearised_state_error_dyn} \\
        \mr{A}_t &:= \tD_e |_{\mr{\xi}} \lchart(e)
        \cdot \tD_E |_\id \phi_{\mr{\xi}} (E)
        \cdot \tD_e |_{\mr{\xi}} \Lambda(e, \mr{u})
        \cdot \tD_\varepsilon |_0 \lchart^{-1}(\varepsilon) \label{eq:state_matrix_A_dfn}.
    \end{align}
\end{lemma}

\begin{proof}
The nonlinear pre-observer dynamics of the global state error \eqref{eq:state_error_kinematics} in local coordinates $\varepsilon = \lchart(e)$ are
\begin{align} \label{eq:vareps_nldyn}
    \dot{\varepsilon} &= \tD \lchart \cdot \tD  \phi_{\lchart^{-1}(\varepsilon)} (
    \Lambda(\lchart^{-1}(\varepsilon), \mr{u}) - \Lambda(\mr{\xi}, \mr{u})).
\end{align}
Clearly, $\Lambda(\lchart^{-1}(0), \mr{u}) = \Lambda(\mr{\xi}, \mr{u})$ since $\lchart^{-1}(0) = \mr{\xi}$.
Hence, linearising $\dot{\varepsilon}$ about $\varepsilon = 0$ yields
\begin{align}
    &\tD \lchart \cdot \tD  \phi_{\lchart^{-1}(\varepsilon)} (
    \Lambda(\lchart^{-1}(\varepsilon), \mr{u}) - \Lambda(\mr{\xi}, \mr{u})) \notag \\
    &= \tD \lchart \cdot \tD \phi_{\lchart^{-1}(0)}
    (\Lambda(\lchart^{-1}(0), \mr{u}) - \Lambda(\mr{\xi}, \mr{u}))
\notag \\ & \phantom{=}
    + \tD_\varepsilon |_{0}  \left( \tD \lchart \cdot \tD \phi_{\lchart^{-1}(\varepsilon)} (
    \Lambda(\lchart^{-1}(\varepsilon), \mr{u}) - \Lambda(\mr{\xi}, \mr{u}) \right) [\varepsilon]
    + O(\vert \varepsilon \vert^2), \label{eq:A_lin_1} \\
    &=
    \tD_\varepsilon |_{0} \left( \tD \lchart \cdot \tD \phi_{\lchart^{-1}(\varepsilon)} \right)[\varepsilon] \cdot \left(
    \Lambda(\lchart^{-1}(0), \mr{u}) - \Lambda(\mr{\xi}, \mr{u}) \right)
\notag \\ & \phantom{=}
    + \tD \lchart \cdot \tD \phi_{\lchart^{-1}(0)} \cdot \tD_\varepsilon |_{0}  \left(
    \Lambda(\lchart^{-1}(\varepsilon), \mr{u}) - \Lambda(\mr{\xi}, \mr{u}) \right) [\varepsilon]
    + O(\vert \varepsilon \vert^2), \label{eq:A_lin_2} \\
    &= \tD \lchart
    \cdot \tD \phi_{\mr{\xi}}
    \cdot \tD_e |_{\mr{\xi}} \Lambda(e, \mr{u})
    \cdot \tD \lchart^{-1} [\varepsilon]
    + O(\vert \varepsilon \vert^2), \label{eq:A_lin_3} \\
    &= \mr{A}_t \varepsilon
    + O(\vert \varepsilon \vert^2). \notag
\end{align}
Here \eqref{eq:A_lin_1} is the first order Taylor expansion in local coordinates about $\varepsilon = 0$.
Equation \eqref{eq:A_lin_2} follows since $\Lambda(\lchart^{-1}(0), \mr{u}) - \Lambda(\mr{\xi}, \mr{u}) = 0$ and by expanding the first-order term using the product rule.
Equation \eqref{eq:A_lin_3} is the result of applying $\Lambda(\lchart^{-1}(0), \mr{u}) - \Lambda(\mr{\xi}, \mr{u}) = 0$ to eliminate the first term, while noting that $\tD_\varepsilon |_{0} \Lambda(\mr{\xi}, \mr{u}) \equiv 0$ and applying the chain rule to simplify the second term in \eqref{eq:A_lin_2}.
The final line follows from \eqref{eq:state_matrix_A_dfn}.
\end{proof}

The primary role of the error dynamics linearisation is in the covariance propagation instantiated in the Ricatti equation \eqref{eq:eqf_riccati} where the covariance is propagated along the pre-observer trajectories \cite{2021_mahony_arcar}.
Barrau and Bonnabel \cite{2017_Barrau_tac} showed that,
if the manifold $\calM = \grpG$, the origin is chosen $\mr{\xi} = \id$, the local coordinates are chosen $\lchart(e) := \log(e)$, and the system dynamics are `group affine' (cf.~Appendix \ref{sec:iekf_specialisation}), then the pre-observer dynamics are exact,
$\dot{\varepsilon} = \mr{A}_t \varepsilon$.
This \emph{exact} linearisation of the pre-observer error dynamics removes the $O(\vert \varepsilon \vert^2)$ linearisation error and significantly improves the performance of the filter by reducing linearisation error in the Ricatti equation.
Even without the group affine property, the equivariant structure of the error dynamics provides significant advantages over the standard EKF.
The error dynamics are linearised at a single point $\mr{\xi}$ and using a single coordinate chart $\lchart$, which can be designed intentionally to minimise the $O(\vert \varepsilon \vert^2)$ linearisation error.
As we show in Lemma \ref{lem:equivariant_output}, it is also possible to exploit the equivariant system structure to reduce linearisation error in the output approximation, further improving the filter performance.


\subsubsection{System Output}

Consider the state error $e$ \eqref{eq:state_error_dfn}, and let $\xi \in \calM$ and $\hat{X} \in \grpG$ denote the true system state and observer state respectively.
Let $\varepsilon \in \R^m$ represent local coordinates for $e$ as in \eqref{eq:linearised_state_error_dfn}.
The output $y = h(\xi)$ can be written
\begin{align} \label{eq:h_in_error_coords}
    h(\xi) = h(\phi(\hat{X}, e)) = h(\phi_{\hat{X}}(\lchart^{-1}(\varepsilon))).
\end{align}
Note that substituting $\varepsilon = 0$ gives
\begin{align*}
h(\phi_{\hat{X}}(\lchart^{-1}(0))) & =
h(\phi_{\hat{X}}(\mr{\xi}))
= h(\hat{\xi}).
\end{align*}
In common with error state Kalman filters, the output $y = y(\varepsilon)$ is considered as a function of the error \eqref{eq:h_in_error_coords} while the predicted output $\hat{y} = h(\hat{\xi})$ is considered an independent signal.
Define the output residual
\begin{align}
    \tilde{y} &= y(\varepsilon) - \hat{y}.
\label{eq:y_tilde}
\end{align}
Linearising $\tilde{y}$ as a function of $\varepsilon \in \R^m$ around $\varepsilon = 0$ yields
\begin{align}
\tilde{y} & = C_t \varepsilon + O(\vert \varepsilon \vert^2), \label{eq:linearised_output_error_dyn_NEW} \\
C_t &:= \tD_\xi |_{\hat{\xi}} h (\xi)
\cdot \tD_e |_{\mr{\xi}} \phi_{\hat{X}} (e)
\cdot \tD_\varepsilon |_0 \lchart^{-1}(\varepsilon).
\label{eq:output_matrix_C_dfn}
\end{align}
The matrix $C_t \in \R^{n \times m}$, the Jacobian of \eqref{eq:h_in_error_coords}, is termed the \emph{standard output matrix}.

\subsubsection{Equivariant Output Linearisation}

When the system has output equivariance, the linearisation of the output function can be improved to obtain $O(\vert \varepsilon \vert^3)$ error.
A filter's ability to incorporate information from measurements correctly is fundamental to its robustness and transient performance.
As shown in the simulation results in \S\ref{sec:sphere_example}, the equivariant output linearisation presented below greatly improves the EqF performance both with and without the presence of noise in the measurement signals.

\begin{lemma} \label{lem:equivariant_output}
Suppose the local coordinates $\lchart : \calU_{\mr{\xi}} \subset \calM \to \R^m$ are normal coordinates of the Lie group about $\mr{\xi}$.
Then
\begin{align}
    \tilde{y} &= C^\star_t \varepsilon + O(\vert \varepsilon \vert^3), \\
    C^\star_t \varepsilon &= \frac{1}{2} \left(\tD_{E | \id} \rho(E, y) + \tD_{E | \id} \rho(E, \hat{y}) \right) \Ad_{\hat{X}^{-1}} \varepsilon^\wedge,
    \label{eq:equivariant_output_matrix}
\end{align}
where $\cdot^\wedge : \R^m \to \gothm \subset \gothg$ is the identification of $\gothm$ with $\R^m$ used in defining $\lchart$.
\end{lemma}

\begin{proof}
By construction $\lchart^{-1}(\varepsilon) =  \phi_{\mr{\xi}} (\exp( \varepsilon^\wedge)) = e$.
Recalling \eqref{eq:h_in_error_coords} and \eqref{eq:y_tilde} one has
\begin{align*}
    \tilde{y}(\varepsilon; \hat{X})
    &= h(\phi(\hat{X}, \phi(\exp(\varepsilon^\wedge), \mr{\xi}))) - h(\phi(\hat{X}, \mr{\xi})).
\end{align*}
Clearly, $\tilde{y}(0; \hat{X}) = 0$.
Using the equivariance of $h$, one has
\begin{align*}
    h(\phi(\hat{X}, \phi(\exp( \varepsilon^\wedge),\mr{\xi})))
    & = h(\phi(\exp(\varepsilon^\wedge) \hat{X}, \mr{\xi})), \\
    &= h(\phi(\hat{X} \hat{X}^{-1} \exp(\varepsilon^\wedge) \hat{X}, \mr{\xi})), \\
    &= h(\phi(\exp(\Ad_{\hat{X}^{-1}} (\varepsilon^\wedge)) ,\phi_{\hat{X}}(\mr{\xi}))), \\
    &= \rho(\exp(\Ad_{\hat{X}^{-1}} (\varepsilon^\wedge)), h(\hat{\xi})).
\end{align*}
Setting $\hat{y} = h(\hat{\xi})$ and differentiating $\tilde{y}$ at $\varepsilon = 0$ in a direction $\gamma \in \R^m$ yields
\begin{align*}
    \tD_{x | 0} \tilde{y}(x; \hat{X}) [\gamma]
    &= \tD_{E | \id} \rho_{\hat{y}}(E) \Ad_{\hat{X}^{-1}} \gamma^\wedge.
\end{align*}
Although this formula allows for an arbitrary $\gamma \in \R^m$, the linearisation is computed for $\gamma = \varepsilon$.

The fact that the linearisation is computed in the same direction $\varepsilon$ as the coordinates of the error can be exploited along with equivariance to obtain a second linearisation point.
In particular, we will compute the differential of $y$ at $\varepsilon = \lchart(e)$ \emph{in direction} $\varepsilon$.
Note that $\phi_{\mr{\xi}}(\exp(\varepsilon^\wedge)) = \phi(\hat{X}^{-1}, \xi)$, and therefore $\phi_{\mr{\xi}}(\exp(\varepsilon^\wedge) \hat{X}) = \xi$.
Then
\begin{align}
    \tD_{x | \varepsilon} \tilde{y}(x; \hat{X}) [\varepsilon]
    &= \left. \ddt \right\vert_{t=0} h(\phi(\exp((1+t)\varepsilon^\wedge) \hat{X}, \mr{\xi})), \notag\\
    &= \left. \ddt \right\vert_{t=0} h(\phi(\exp(\varepsilon^\wedge) \exp(t \varepsilon^\wedge) \hat{X}, \mr{\xi})), \notag\\
    &= \left. \ddt \right\vert_{t=0} h(\phi(\exp(\varepsilon^\wedge) \hat{X} \hat{X}^{-1} \exp(t \varepsilon^\wedge) \hat{X}, \mr{\xi})), \notag\\
    &= \left. \ddt \right\vert_{t=0} h(\phi(\hat{X}^{-1} \exp(t \varepsilon^\wedge) \hat{X}, \xi)), \notag\\
    &= \left.\ddt \right\vert_{t=0} \rho(\exp(t \Ad_{\hat{X}^{-1}} \varepsilon^\wedge), h(\xi)), \notag\\
    &= \tD_{E | \id} \rho(E, y) \Ad_{\hat{X}^{-1}} \varepsilon^\wedge .\notag
\end{align}
Although the differential is posed at the unknown error state $\varepsilon$, it is evaluated using only the known measurement data $y$.

Consider the Taylor expansion of the differential $\tD_{x | \varepsilon} \tilde{y}(x; \hat{X})$ with respect to $\varepsilon$ around $\varepsilon = 0$:
\begin{align*}
    \tD_{x | \varepsilon} \tilde{y}(x; \hat{X}) [\cdot]
    = \tD_{x | 0} \tilde{y}(x; \hat{X})[\cdot] + \tD^2_{x | 0} \tilde{y}(x; \hat{X}) [\varepsilon, \cdot] + O(\vert \varepsilon \vert^2),
\end{align*}
and hence
\begin{align*}
    \tD^2_{x | 0} \tilde{y}(x; \hat{X}) [\varepsilon, \varepsilon]
    = \tD_{x | \varepsilon} \tilde{y}(x; \hat{X}) [\varepsilon]
    - \tD_{x | 0} \tilde{y}(x; \hat{X})[\varepsilon]
    + O(\vert \varepsilon \vert^3).
\end{align*}

The result \eqref{eq:equivariant_output_matrix} follows from taking the Taylor expansion of $\tilde{y}$ with respect to $\varepsilon$ and substituting
\begin{align*}
    \tilde{y}(\varepsilon; \hat{X})
    &= \tilde{y}(0; \hat{X})
    + \tD_{x | 0} \tilde{y}(x; \hat{X}) [\varepsilon]
    + \frac{1}{2}\tD^2_{x | 0} \tilde{y}(x; \hat{X}) [\varepsilon, \varepsilon]
    + O(\vert \varepsilon \vert^3), \\
    &= \frac{1}{2}(\tD_{x | \varepsilon} \tilde{y}(x; \hat{X}) [\varepsilon]
    + \tD_{x | 0} \tilde{y}(x; \hat{X})[\varepsilon])
    + O(\vert \varepsilon \vert^3), \\
    &= \frac{1}{2}(\tD_{E | \id} \rho(E, y)
    + \tD_{E | \id} \rho(E, \hat{y}) ) \Ad_{\hat{X}^{-1}} \varepsilon^\wedge
    + O(\vert \varepsilon \vert^3).
\end{align*}
\end{proof}

\section{Equivariant Filter (EqF)}
\label{sec:eqf_definition}

Consider a kinematic system \eqref{eq:kinematic_system_function} with a state symmetry $\phi: \grpG \times \calM \to \calM$.
Assume the system is equivariant with input symmetry $\psi: \grpG \times \vecL \to \vecL$ and has an equivariant lift $\Lambda : \calM \times \vecL \to \gothg$.
Let $\xi \in \calM$ denote the true state of the system, with trajectory determined by the measured input $u \in \vecL$.
Denote the configuration output $y = h(\xi)$.

We construct the Equivariant Filter (EqF) as follows.
Let $\hat{X} \in \grpG$ denote the observer state.
Pick an arbitrary fixed origin $\mr{\xi} \in \calM$.
For a general output map, set $C_t$ to be the standard output matrix defined in \eqref{eq:linearised_output_error_dyn_NEW}.
If the system has output equivariance then set
$C_t = C^\star_t$ to be the equivariant output matrix as defined in \eqref{eq:equivariant_output_matrix}.
In this case the resulting algorithm is termed the EqF$^\star$. 
Let $\mr{A}_t$ denote the state matrix as defined in \eqref{eq:state_matrix_A_dfn}.
Choose an initial value for the Riccati term $\Sigma_0 \in \Sym_+(m)$, where $\Sym_+(m)$ is the set of positive-definite symmetric $m \times m$ matrices, and pick a state gain matrix $M_t \in \Sym_+(m)$ and an output gain matrix $N_t \in \Sym_+(n)$.
Choose a right inverse $\tD_E |_\id \phi_{\mr{\xi}}(E)^\dag$ of $\tD_E |_\id \phi_{\mr{\xi}}(E)$; that is, $\tD_E |_\id \phi_{\mr{\xi}}(E) \cdot \tD_E |_\id \phi_{\mr{\xi}}(E)^\dag = \id$.

The proposed equivariant filter is given by the solution of
\begin{align}
    \dot{\hat{X}} &= \tD L_{\hat{X}} \Lambda(\phi(\hat{X}, \mr{\xi}), u) + \tD R_{\hat{X}} \Delta, \quad \hat{X}(0) = \id, \label{eq:eqf_group_observer} \\
    \Delta &= \tD_E |_\id \phi_{\mr{\xi}}(E)^\dag \tD \lchart^{-1} \Sigma C_t^\top N_t^{-1} (y - h(\phi(\hat{X}, \mr{\xi}))), \label{eq:eqf_delta_correction} \\
    \dot{\Sigma} &= \mr{A}_t \Sigma + \Sigma {\mr{A}_t}^\top + M_t - \Sigma C_t^\top N_t^{-1} C_t \Sigma, \quad \Sigma(0) = \Sigma_0\label{eq:eqf_riccati}.
\end{align}

If the pair $(\mr{A}_t, C_t)$ is uniformly observable in the sense of Proposition 1 of \cite{2017-Morin-geometric}, then $\Sigma(t)$ is bounded above and below, and the Riccati equation \eqref{eq:eqf_riccati} is well-defined for all time \cite{2001_deylon_tac}.

Provided that the error trajectory $\varepsilon(t)$ remains well-defined for all time $t \geq 0$, \cite[Theorem 1.1.1]{2003_krener_ekf} provides sufficient conditions for the convergence of $\varepsilon \to 0$.
In particular, if the error system is uniformly observable, the second derivative of the error dynamics is bounded, and the initial error is sufficiently small, then the error $\varepsilon$ and the Lyapunov function
\begin{align} \label{eq:linearised_lyapunov_func}
    \Lyap(t) := \varepsilon^\top \Sigma^{-1} \varepsilon,
\end{align}
converge exponentially to zero as $t \to \infty$ \cite{2003_krener_ekf}.

\begin{remark}
Recent work by the authors \cite{2013_zamani_tac,2015_saccon_tac,2021_mahony_arcar} showed that the Ricatti equation \eqref{eq:eqf_riccati} can be augmented by a curvature modification that compensates for the reset process, where the linearisation point is continually translated to track the observer state,  in the extended Kalman filter derivation \cite{2021_mahony_arcar}.
Several works have shown that, during the transient at least, an appropriate curvature modification term can improve filter performance \cite{2015_zamani_arxiv,2021_mahony_arcar}.
Curvature is directly connected to parallel transport on a manifold and is an additional structure that can be chosen independently from the homogeneous space structure.
One possible choice is to define the normal coordinates to be flat, that is, parallel transport on the manifold is just translation in local coordinates.
The Ricatti equation \eqref{eq:eqf_riccati} corresponds to this choice.
Such a choice has the advantage of simplicity; however, the associated affine connection will usually have non-zero torsion.
The relative benefit or consequence of choosing different geometries, with different curvatures, along with symmetry or torsion of the associated connections remains an open question in equivariant systems theory.
\end{remark}

\subsection{EqF Gain Tuning}
\label{sec:gain_tuning}

The choice of gain matrices $\Sigma_0$, $M_t$ and $N_t$ can greatly influence the performance of the EqF.
In the context of a Kalman-Bucy filter, $\Sigma_0$ reflects the uncertainty in the initial state estimate, and $M_t$ and $N_t$ are optimally chosen to be the intensities (covariances) of zero-mean Gaussian noise terms added to the filter dynamics and output, respectively.
Similar choices can be made for the EqF.
The initial value of the Riccati term $\Sigma_0$ can be chosen to reflect the uncertainty in the initial state estimate as expressed in the chosen local coordinates.
Suppose the measured velocity $\td u_m = \td u + \td \mu_u$ and the measured output $\td y_m = \td y + \td \nu_y$, where $\td \mu_u \sim \mathbf{W}(0, M_t^m)$ and $\td \nu_y \sim \mathbf{W}(0, N_t^m)$ are Wiener processes.
Then, re-linearising the pre-observer error dynamics \eqref{eq:state_error_kinematics} with $\Delta \equiv 0$ and output \eqref{eq:h_in_error_coords} yields
\begin{align}
    \td \varepsilon &= \mr{A}_t \varepsilon \td t + B_t \td \mu_u, \\
    \td \tilde{y} &= C_t \varepsilon \td t + \td \nu_y,
\end{align}
where the \emph{input matrix} $B_t$ is obtained by linearising the pre-observer error dynamics with respect to a perturbation of the measured input,
\begin{align}
    B_t &:= \tD_e |_{\mr{\xi}} \lchart(e)
    \cdot \tD_E |_\id \phi_{\mr{\xi}} (E)
    \cdot \Ad_{\hat{X}}
    \cdot \tD_u |_{u_m} \Lambda(\hat{\xi}, u), \label{eq:B_matrix_dfn} \\
    &= \tD_e |_{\mr{\xi}} \lchart(e)
    \cdot \tD_E |_\id \phi_{\mr{\xi}} (E)
    \cdot \tD_w |_{\mr{u}_m} \Lambda(\mr{\xi}, w)
    \cdot \psi_{\hat{X}^{-1}}.
    \notag
\end{align}
Based on this formulation, the EqF gain matrices can be chosen by
\begin{align}
    M_t &= M_\varepsilon + B_t M_t^m B_t^\top, \\
    N_t &= N_\varepsilon + N_t^m.
\end{align}
The matrices $M_\varepsilon \in \Sym_+(m)$ and $N_\varepsilon \in \Sym_+(n)$ are optimally set to zero in the case of a linear system, but can otherwise be used by the practitioner to model the error introduced to the dynamics and output by linearisation.

\section{Example Revisited: Single Bearing Estimation}
\label{sec:sphere_example}

\newhl{
The system described in Section \ref{sec:motivating_example} of bearing estimation on the sphere provides an illustrative example of a system on a homogeneous space where the EqF design methodology may be applied.
For additional examples of EqF applications, we refer the reader to \cite{2021_goor_EquivariantFilterVisual,2020_vangoor_cdc,2021_mahony_arcar}.
}

\subsection{Equivariant System}

Here we describe the design preliminaries for the EqF and EqF$^\star$ for the single bearing estimation problem following Algorithm \ref{alg:eqf_design_setup} as described in Appendix \ref{sec:filter_implementation}.

\subsubsection{State Symmetry}
Consider the Lie group of 3D rotations
\[
    \SO(3) = \{ R \in \R^{3 \times 3} \; | \; R^\top R = I_3, \; \det(R) = 1 \}.
\]
This group has a right action on the sphere $\phi : \SO(3) \times \Sph^2 \to \Sph^2$ given by
\begin{align}
    \phi(R, \eta) := R^\top \eta.
\end{align}

\subsubsection{System Equivariance}
Define the map $\psi : \SO(3) \times \R^3 \to \R^3$ to be
\begin{align}
    \psi(R, \Omega) := R^\top \Omega.
\end{align}
Then the system \eqref{eq:ex_sphere_dynamics} is equivariant with respect to the state action $\phi$ and input action $\psi$.
To see this, let $\Omega \in \R^3$, $\eta \in \Sph^2$, and $R \in \SO(3)$ be arbitrary, and compute
\begin{align*}
    \Phi(R, f_{\Omega}) (\eta)
    &= \tD \phi_R (-\Omega^\times R \eta), \\
    &= -R^\top \Omega^\times R \eta, \\
    &= - (R^\top \Omega)^\times \eta, \\
    &= f_{\psi(R, \Omega)}(\eta).
\end{align*}

\subsubsection{Equivariant Lift}
The Lie algebra $\so(3)$ of $\SO(3)$ can be written as the subspace of skew-symmetric matrices,
\begin{align*}
    \so(3) &:= \{ \omega \in R^{3 \times 3} \; | \; \omega^\top = -\omega \}.
\end{align*}
Define the candidate lift function $\Lambda : \Sph^2 \times \R^3 \to \so(3)$ by
\begin{align}
    \Lambda(\eta, \Omega) := \Omega^\times
\end{align}
To check that it is indeed a lift, evaluate the lift condition \eqref{eq:lift_condition}.
\begin{align*}
    \tD_R |_\id \phi_\eta(R)\left[ \Lambda(\eta, \Omega) \right]
    &= \tD_R |_\id \phi_\eta(R)\left[ \Omega^\times \right], \\
    &= (\Omega^\times)^\top \eta, \\
    &= - \Omega^\times \eta, \\
    &= f_{\Omega}(\eta),
\end{align*}
as required.
Next, check the equivariance of $\Lambda$ as in \eqref{eq:equivariant_lift}.
\begin{align*}
    \Lambda(\phi(R, \eta), \psi(R, \Omega))
    &= \Lambda(R^\top \eta, R^\top \Omega), \\
    &= (R^\top \Omega)^\times, \\
    &= R^\top \Omega^\times R, \\
    &= \Ad_{R^{-1}} \Lambda(\eta, \Omega),
\end{align*}
as required.

\subsubsection{Output Equivariance}
The action $\rho : \SO(3) \times \R^3 \to \R^3$ given by
\begin{align}
    \rho(R, y) = R^\top y,
\end{align}
ensures the system \eqref{eq:ex_sphere_output} has output equivariance since
\begin{align*}
    \rho(R, h(\eta))
    = R^\top (c_m \eta)
    = c_m R^\top \eta
    = h(\phi(R, \eta)).
\end{align*}

\subsubsection{Origin and State Error}
Fix the origin element $\mr{\eta} = \eb_1 \in \Sph^2$, and let $\hat{R} \in \SO(3)$ denote the observer state.
The global state error is given by
\begin{align*}
    e = \phi(\hat{R}^{-1}, \eta) = \hat{R} \eta
\end{align*}
where $\eta \in \Sph^2$ is the true state of the system.

Since the system exhibits output equivariance, we choose to use normal coordinates.
Explicitly, define
\begin{align*}
    \gothm &:= \{ v^\wedge \in \R^{3 \times 3} \; | \; v \in \R^2 \} \subset \so(3), \\
  (v_2, v_3)^\wedge &:= (0, v_2, v_3)^\times,
\end{align*}
where the indices $(v_2,v_3) \in \R^2$ are chosen to correspond to the associated indices for the embedding into $\so(3)$.
Then the normal coordinates for $\Sph^2$ about $\eb_1$ are given by
\begin{align}
    \lchart(e) &:= -\text{atan2}(\vert \eb_1 \times e \vert, \eb_1^\top e)  \begin{pmatrix}
        0_{2 \times 1} & I_2
    \end{pmatrix}  \frac{\eb_1 \times e}{\vert \eb_1 \times e \vert}, \label{eq:sphere_normal_coordinates} \\
    \lchart^{-1}(\varepsilon) &:= \phi(\exp(\varepsilon^\wedge), \eb_1). \notag
\end{align}

In this system, $\tD_E |_\id \phi_{\eb_1}(E)$ is not invertible,
\begin{align*}
    \tD_{E |_\id} \phi_{\eb_1}(E)[\omega^\times]
    &= \left. \ddt \right\vert_{t=0} \phi(\exp(t \omega^\times), \eb_1), \\
    &= - \omega^\times \eb_1, \\
    &= \eb_1^\times \omega.
\end{align*}
We propose the following right inverse (required in
\eqref{eq:eqf_delta_correction}),
\begin{align*}
    \tD_E |_\id \phi_{\eb_1}(E)^\dagger [u]
    := (u_3, -u_2)^\wedge
    \in \gothm,
\end{align*}
where $u = (0, u_2, u_3)^\top \in \tT_{\eb_1} \Sph^2$ is an arbitrary tangent vector in the embedded coordinates for $\Sph^2$ at $\eb_1$.
To see that this indeed defines a right-inverse, compute
\begin{align*}
    \tD_{E |_\id} \phi_{\eb_1}(E) \tD_E |_\id \phi_{\eb_1}(E)^\dagger [u]
    &= \eb_1^\times \begin{pmatrix}
        0 & u_3 & -u_2
    \end{pmatrix}^\top, \\
    &= \begin{pmatrix}
        0 & u_2 & u_3
    \end{pmatrix}^\top, \\
    &= u.
\end{align*}

\subsubsection{EqF Matrices}
The EqF matrices are obtained by specialising the general matrix formulas to the specific example.
The state matrix $\mr{A}_t$ \eqref{eq:linearised_state_error_dyn}, the input matrix $B_t$ \eqref{eq:B_matrix_dfn}, the standard output matrix $C_t$ \eqref{eq:output_matrix_C_dfn}, and the equivariant output matrix $C_t = C^\star_t$ \eqref{eq:equivariant_output_matrix} are given by
\begin{align*}
    \mr{A}_t &= 0_{2 \times 2}, &
    B_t &= \begin{pmatrix}
        0_{2 \times 1} & I_2
    \end{pmatrix} \hat{R} , \\
    C_t &= \hat{y}^\times \hat{R}^\top \begin{pmatrix}
        0_{1 \times 2} \\ I_2
    \end{pmatrix}, &
    C^\star_t &= \frac{1}{2}\left( y^\times + \hat{y}^\times \right) \hat{R}^\top \begin{pmatrix}
        0_{1 \times 2} \\ I_2
    \end{pmatrix}.
\end{align*}
\newhl{
The state matrix $\mr{A}_t$ and output matrix $C_t$ can be compared to the EKF matrices \eqref{eq:ex_sphere_dynamics} and \eqref{eq:EKF_C1_t}, respectively, derived in \S\ref{sec:motivating_example}.
}

\subsection{EqF Implementation}

We implement the EqF equations (\ref{eq:eqf_group_observer}-\ref{eq:eqf_riccati}) as detailed in
Algorithm \ref{alg:eqf_design_methodology} in Appendix \ref{sec:filter_implementation}.

The observer dynamics are given by specialising \eqref{eq:eqf_group_observer},
\begin{align*}
    \dot{\hat{R}}
    &= \tD L_{\hat{R}} \Lambda( \phi(\hat{R}, \eb_1), \Omega)
    + \tD R_{\hat{R}} \Delta, \\
    &= \hat{R} \Omega^\times
    + \Delta \hat{R},
\end{align*}
where the correction term $\Delta$ is computed according to \eqref{eq:eqf_delta_correction}.

\subsection{Simulation Results}

To verify the observer design for this example, we performed a simulation of a robot rotating with an angular velocity \newhl{ $\Omega(t) = (0.1\cos(2 t), 0.2\sin(t), -0.1 \cos(1.5 t))$~rad/s}, where $t$ is the simulation time in seconds.
The initial state, the measured angular velocity, and the measured output were chosen by
\begin{align}
    \eta(0) &= \frac{\eb_1 + \mu_0}{\vert \eb_1 + \mu_0 \vert}, &
    \mu_0 &\sim N(0, 2.0^2 I_3), \notag \\
    \Omega_m &= \Omega + \mu_u, &
    \mu_u &\sim N(0, 0.01^2 I_3), \notag \\
    y_m &= \eta + \nu_y, &
    \nu_y &\sim N(0, 0.05^2 I_3),
    \label{eq:example_noise}
\end{align}
respectively.
The state $\eta(t)$ was then computed by integrating
\[
    \ddt \eta(t) = f_{\Omega(t)}(\eta(t)) = - \Omega(t)^\times \eta(t).
\]
The EqF gain matrices were chosen according to the procedure outlined in \S\ref{sec:gain_tuning}.

We also implemented an extended Kalman filter (EKF) as described in \S\ref{sec:motivating_example} to compare its performance to that of the EqF.
The system (\ref{eq:ex_sphere_dynamics},\ref{eq:ex_sphere_output}), the EqF equations (\ref{eq:eqf_group_observer}-\ref{eq:eqf_riccati}), and the EKF were all implemented in python3 and integrated for $5.0$~s using Euler integration with a time step of $0.01$~s.
Both the EqF and EKF were given the initial estimate $\hat{\eta}(0) = \eb_1$.

In order to verify the local exponential convergence of the proposed filters in the absence of noise, we performed a simulation with the gyroscope and magnetometer noise set to zero, that is, $\mu_u = 0, \nu_y = 0$.
\newhl{
Figure \ref{fig:noiseless_results} shows the absolute angle $\tilde{\theta}$ between the estimated direction and true direction and the Lyapunov value \eqref{eq:linearised_lyapunov_func} for each filter, where
\begin{align}
    \label{eq:absolute_angle_error}
    \tilde{\theta} := \arccos(\vert \hat{\eta}^\top \eta \vert).
\end{align}
The results demonstrate the performance of each filter under ideal conditions where the only error is due to the difference between the initial state $\eta(0)$ and the initial estimate $\eb_1$.
That is, the true initial bearing was drawn from the distribution described in \eqref{eq:example_noise}, but the measurements were taken to be $\Omega_m = \Omega$ and $y_m = y$ exactly.
It is clear to see that the EqF is locally exponentially convergent, and that the EqF$^\star$ exhibits faster initial convergence.}

\begin{figure}[!htb]
    \centering
    \includegraphics[trim={0 0.2cm 0 0.8cm},clip,width=\figwidth\linewidth]{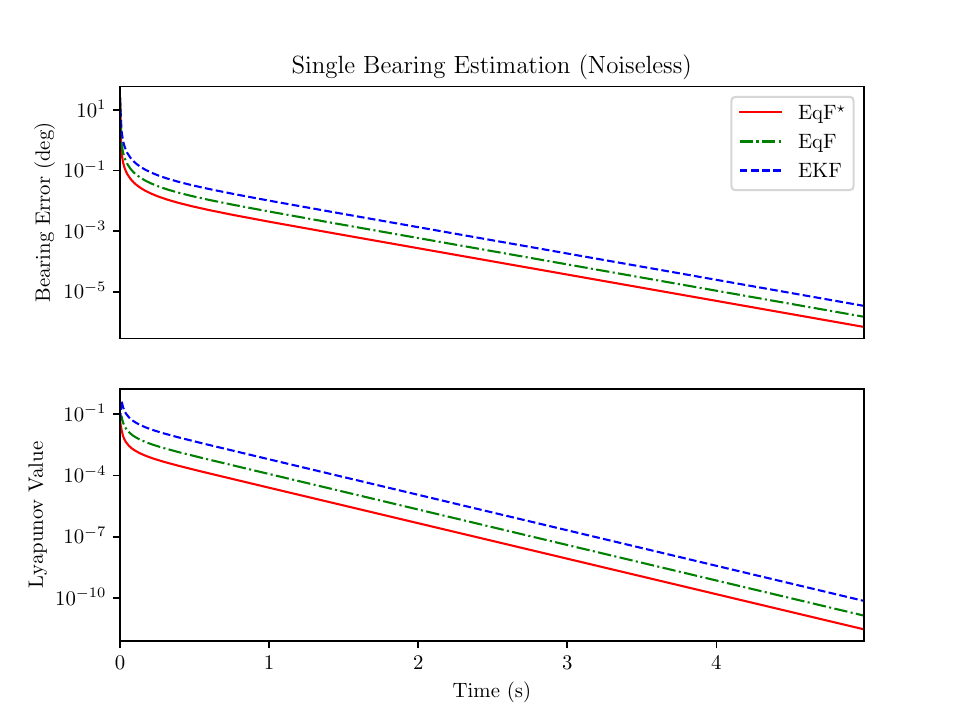}
    \caption{The angle error \eqref{eq:absolute_angle_error} and Lyapunov value \eqref{eq:linearised_lyapunov_func} for each of the filters without noise added to any of the gyroscope or magnetometer measurements.
    The EqF$^\star$ (red solid line) shows faster initial convergence than both the EqF (green dot-dashed line) and the EKF (blue dashed line).
    }
    \label{fig:noiseless_results}
\end{figure}

We also performed 500 monte carlo simulations with noise added to the velocity, measurement, and initial conditions, generated according to \eqref{eq:example_noise}.
Figure \ref{fig:noisy_results} shows the distribution of the absolute angle error \eqref{eq:absolute_angle_error} between the estimated direction and true direction for each filter, as well as the Lyapunov value \eqref{eq:linearised_lyapunov_func} for each of the filters.

\begin{figure}[!htb]
    \centering
    \includegraphics[trim={0 0.2cm 0 0.8cm},clip,width=\figwidth\linewidth]{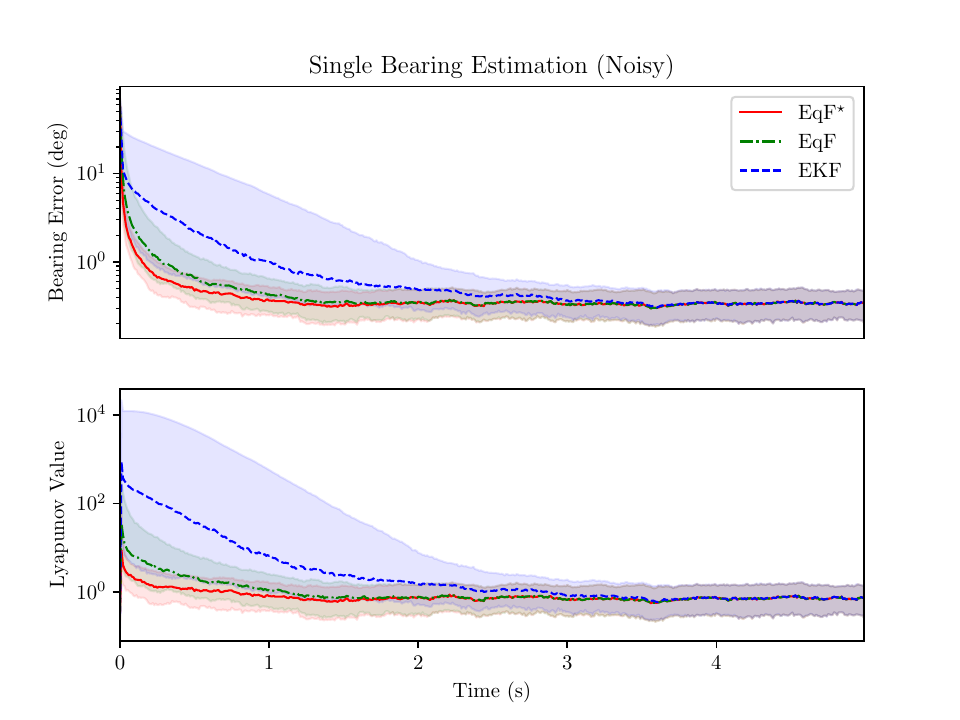}
    \caption{The median angle error \eqref{eq:absolute_angle_error} and Lyapunov value \eqref{eq:linearised_lyapunov_func} for each of the filters over 500 trials with noise generated for each trial independently.
    The EqF$^\star$ (red solid line) outperforms both the EqF (green dot-dashed line) and the EKF (blue dashed line) in terms of both angle error and Lyapunov value.
    The coloured areas show the 25th and 75th percentile for each filter's angle error and Lyapunov value.
    }
    \label{fig:noisy_results}
\end{figure}

Figure \ref{fig:output_linearisation} shows the error introduced by different linearisations of the output.
Each point $\eta \in \Sph^2$ other than $-\eb_1$ has been mapped to $\R^2$ using spherical coordinates, and each heatmap shows, for a given filter, the absolute difference between the true measurement residual $\tilde{y} = \eta - h(\eb_1)$ and the linearised measurement residual $C_t \lchart(\eta)$.
Note that $\lchart(\eta) := \eta - \eb_1$ in the case of the EKF.
The EqF and EKF show comparable performance, with the EKF having slightly higher error near $\eb_1$ and the EqF having slightly higher error further away from $\eb_1$.
The output linearisation used by the EqF$^\star$ is superior to the other filters, both near $\eb_1$ as well as far from it, as expected from Lemma \ref{lem:equivariant_output}.

\begin{figure}[!htb]
    \centering
    \includegraphics[trim={0 0.2cm 0 0.6cm},clip,width=\figwidth\linewidth]{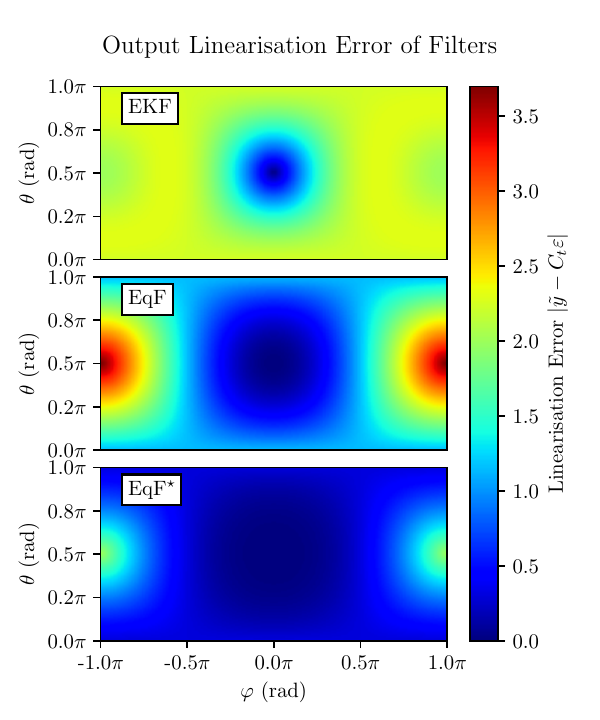}
    \caption{The error of linearising the output residual $\tilde{y}$ for the EKF, EqF, and EqF$^\star$.
    The EKF and EqF are similar in terms of overall quality, with the EqF performing better close to the linearisation point.
    The EqF$^\star$ shows clearly superior performance to both the EKF and EqF as expected from Lemma \ref{lem:equivariant_output}.
    }
    \label{fig:output_linearisation}
\end{figure}


The example system shown is of interest for several reasons.
First, the system is defined on a homogeneous space of a Lie group rather than on a Lie group itself, necessitating the development of a lifted system to apply equivariant observer design methods.
This also precludes the application of the popular IEKF \cite{2017_Barrau_tac} as it is exclusively defined for group affine systems on Lie groups and not on the more general class of equivariant systems on homogeneous spaces.
Second, the system has a symmetry compatible with the output function, enabling the use of Lemma \ref{lem:equivariant_output} to improve the output linearisation.
Figure \ref{fig:output_linearisation} shows that the improved output linearisation results in a significant reduction in linearisation error over the whole space, and Figures \ref{fig:noiseless_results} and \ref{fig:noisy_results} show clearly the positive effect of this improvement on filter performance.

Overall, the simulations demonstrate clearly that the EqF$^\star$ outperforms the EqF, which in turn outperforms the EKF.
The relative margin of improvement may appear small; however, it must be understood in the context of a problem chosen for simplicity and the implementation of an EKF that was designed carefully to exploit the specific structure available, specifically using embedding coordinates in order to obtain linear time varying state dynamics \eqref{eq:ex_sphere_dynamics}.
Conversely, the EqF and EqF$^\star$ were derived following the standard methodology outlined in Appendix \ref{sec:filter_implementation}.
For more complex problems, where there is no longer structure that can be exploited to design a clever EKF, the relative performance gap is expected to increase.

\section{Conclusion}
\label{sec:conclusion}

The Equivariant Filter (EqF) proposed in this paper is a nonlinear observer that exploits symmetry properties of equivariant kinematic systems posed on homogeneous spaces.
The key contributions of this paper are
\begin{itemize}
    \item Proposing a general filter design for equivariant systems with linearised dynamics about a fixed origin point rather than the time-varying state estimate.
    \item Demonstrating how output equivariance leads to an approximation of the output map for the EqF that has third order linearisation error (rather than the usual quadratic linearisation error), improving  filter robustness and transient performance.
    \item Providing a simple example of the EqF application and simulation results where the EqF clearly outperforms the traditional EKF.
\end{itemize}

It is important to note that the example in the present paper is chosen to be as simple as possible.
It may appear that the mathematical overhead of the EqF is not warranted.
Recent works \cite{2021_vangoor_auto,2020_mahony_cdc,2021_goor_EquivariantFilterVisual} provide other examples where the EqF is the only symmetry based filter that can be applied.

The appendices provide details on how to derive and implement the EqF, and the conditions under which the EqF specialises to the invariant extended Kalman filter.

\appendix
\section{Filter Implementation}
\label{sec:filter_implementation}

The following algorithms are used to design the EqF for a given system.
\begin{algorithm}
\caption{EqF Design Preliminaries}
\label{alg:eqf_design_setup}
\begin{enumerate}
    \item Find a Lie group $\grpG$ and state action $\phi : \grpG \times \calM \to \calM$.
    \item Check that the system is equivariant and compute the input symmetry $\psi : \grpG \times \vecL \to \vecL$.
    \item Construct an equivariant lift $\Lambda : \calM \times \vecL \to \gothg$.
    \item Check if there exists an action $\rho : \grpG \times \calN \to \calN$ such that the configuration output is equivariant.
    \item Choose an origin $\mr{\xi} \in \calM$, choose a local coordinate chart $\lchart$ about $\mr{\xi}$, and fix a right-inverse $\tD_X |_\id \phi_{\mr{\xi}}(X)^\dag$ of $\tD_X |_\id \phi_{\mr{\xi}}(X)$.
    \item Initialise the observer state $\hat{X}(0) = \id$ and the Riccati term $\Sigma(0) = \Sigma_0 \in \Sym_+(m)$.
\end{enumerate}
\end{algorithm}

While the filter is presented in continuous time, in practice the equations must be implemented through numerical integration.
Given input and output measurements $u \in \vecL$ and $y = h(\xi) \in \calN$ at a given time, the steps in Algorithm \ref{alg:eqf_design_methodology} are executed.
\begin{algorithm}
\caption{EqF Design Implementation}
\label{alg:eqf_design_methodology}
\begin{enumerate}
    \item Compute the origin velocity $\mr{u} = \psi_{\hat{X}^{-1}}(u)$ and use this to obtain the state matrix $\mr{A}_t$ \eqref{eq:state_matrix_A_dfn} (cf.~Lemma \ref{lem:state_matrix_A_simple}).
    \item Compute the standard output matrix $C_t$ \eqref{eq:output_matrix_C_dfn} or (preferably) the equivariant output matrix $C_t^\star$ \eqref{eq:equivariant_output_matrix}.
    \item Choose state and output gain matrices $M_t \in \Sym_+(m)$ and $N_t \in \Sym_+(n)$.
    \item Update the observer state $\hat{X}(t)$ and Riccati state $\Sigma(t)$ by numerically approximating equations (\ref{eq:eqf_group_observer}-\ref{eq:eqf_riccati}).
\end{enumerate}
\end{algorithm}
For the majority of applications, Euler integration or a higher order Runge-Kutta method is appropriate.
Itazi and Sanyal \cite{2014_izadi_automatica} showed the effectiveness of Lie group variational integrators \cite{2004_leok_foundations} for discretising an observer for invariant attitude dynamics, and a similar approach may also be applied to discretising the EqF for certain systems.

In some cases it may be difficult to compute an explicit algebraic expression for the velocity action $\psi$.
The following Lemma provides a way to implement the EqF without the need to derive $\psi$.

\begin{lemma}
\label{lem:state_matrix_A_simple}
The linearised state matrix $\mr{A}_t$ defined in \eqref{eq:state_matrix_A_dfn} can be written
\begin{align}
    \mr{A}_t
    &= \tD_e |_{\mr{\xi}} \lchart(e)
    \tD_\xi |_{\hat{\xi}} \phi_{\hat{X}^{-1}} (\xi)
    \cdot \tD_E |_\id \phi_{\hat{\xi}} (E) \notag \\
    & \phantom{=} \cdot \tD_\xi |_{\phi_{\hat{X}}(\mr{\xi})} \Lambda(\xi, u)
    \cdot \tD_e |_{\mr{\xi}} \phi_{\hat{X}} (e)
    \cdot \tD_\varepsilon |_0 \lchart^{-1}(\varepsilon).
    \label{eq:state_matrix_A_alt_dfn}
\end{align}
\end{lemma}

\begin{proof}
Recall the equivariant lift condition \eqref{eq:equivariant_lift}.
It follows that
\begin{align*}
    \tD_E & |_\id  \phi_{\mr{\xi}} (E)
    \cdot \tD_e |_{\mr{\xi}} \Lambda(e, \mr{u}) \\
    &= \tD_E |_\id \phi_{\mr{\xi}} (E)
    \cdot \tD_e |_{\mr{\xi}} \Ad_{\hat{X}} \Lambda(\phi_{\hat{X}}(e), u), \\
    &= \tD_E |_\id \phi_{\mr{\xi}} (E) \Ad_{\hat{X}}
    \cdot \tD_e |_{\mr{\xi}} \Lambda(\phi_{\hat{X}}(e), u), \\
    &= \tD_E |_\id \phi_{\mr{\xi}} (E) \Ad_{\hat{X}}
    \cdot \tD_\xi |_{\phi_{\hat{X}}(\mr{\xi})} \Lambda(\xi, u)
    \cdot \tD_e |_{\mr{\xi}} \phi_{\hat{X}} (e), \\
    &= \tD_\xi |_{\hat{\xi}} \phi_{\hat{X}^{-1}} (\xi)
    \cdot \tD_E |_\id \phi_{\hat{\xi}} (E)
    \cdot \tD_\xi |_{\phi_{\hat{X}}(\mr{\xi})} \Lambda(\xi, u)
    \cdot \tD_e |_{\mr{\xi}} \phi_{\hat{X}} (e).
\end{align*}
Then the expression \eqref{eq:state_matrix_A_alt_dfn} follows from the definition of $\mr{A}_t$ in \eqref{eq:state_matrix_A_dfn}.
Unlike \eqref{eq:state_matrix_A_dfn}, the expression \eqref{eq:state_matrix_A_alt_dfn} depends only the measured signal $u \in \vecL$ and not on the origin velocity $\mr{u} = \psi_{\hat{X}^{-1}}(u) \in \vecL$.
\end{proof}

While the definitions of $\mr{A}_t$ in \eqref{eq:state_matrix_A_dfn} and \eqref{eq:state_matrix_A_alt_dfn} are equivalent, they present different challenges in practical implementation of the filter equations.
Using the definition \eqref{eq:state_matrix_A_dfn} requires an explicit algebraic expression for the velocity action $\psi$, which may be challenging to compute.
On the other hand, \eqref{eq:state_matrix_A_alt_dfn} is independent of the algebraic expression of $\psi$, but requires the differentials $\tD_\xi |_{\hat{\xi}} \phi_{\hat{X}^{-1}} (\xi)$ and $\tD_\xi |_{\phi_{\hat{X}}(\mr{\xi})} \Lambda(\xi, u)$ to be recomputed at different state elements $\hat{\xi} \in \calM$ for each iteration.
The expression \eqref{eq:state_matrix_A_alt_dfn} is particularly useful in situations where the equivariant velocity extension of a system is infinite \cite{2021_mahony_EquivariantFilterDesign}.

\section{Specialisation to IEKF}
\label{sec:iekf_specialisation}

The EqF specialises to an invariant extended Kalman filter (IEKF) \cite{2017_Barrau_tac} for a certain subclass of equivariant systems; specifically those with group affine dynamics on a Lie group, where the origin is chosen to be the identity, and where the local coordinates are chosen to be the exponential map in the EqF implementation.

Consider a system function $f : \vecL \to \gothX(\calG)$ where $\calG$ is the torsor of an $m$-dimensional matrix Lie group $\grpG \subset \GL(d)$.
Right translation $R : \grpG \times \calG \to \calG$, defined by $R_X (B) = B X$ is a smooth, transitive right action of $\grpG$ on $\calG$.
Suppose the system is equivariant with respect to $R$ and some velocity action $\psi$, that is,
\begin{align*}
    \tD R_X f_{u}(P) = f_{\psi_X(u)} (P X),
\end{align*}
for all $P \in \calG$, $X \in \grpG$ and $u \in \vecL$.
Define the lift $\Lambda : \calG \times \vecL \to \gothg$ to be
\begin{align*}
    \Lambda(P, u) = P^{-1} f_{u}(P),
\end{align*}
where $P^{-1}$ is understood as a matrix inverse.

Let $P \in \calG$ denote the true state of the system.
Choose the origin element, $\mr{P} = I_d$, to be the identity matrix and let $\hat{X} \in \grpG$ denote the state of the observer, with dynamics
\begin{align*}
    \ddt \hat{X} &:= \hat{X} \Lambda(R_{\hat{X}}(\mr{P}), u) + \Delta \hat{X}, \\
    &= \hat{X} \Lambda(\hat{X}, u) + \Delta \hat{X}, \\
    &= \hat{X} \hat{X}^{-1} f_{u}(\hat{X}) + \Delta \hat{X}, \\
    &= f_{u}(\hat{X}) + \Delta \hat{X}.
\end{align*}
In \cite{2017_Barrau_tac}, Barrau \etal showed (see also \cite[Remark 7.1]{2021_mahony_EquivariantFilterDesign}) that if the system $f$ is ``group affine'', then the dynamics of the error $E := P \hat{X}^{-1}$ depend only on $u$ and not the origin velocity $\mr{u}$,
\begin{align*}
    \dot{E} &= f_{u}(E) - E f_{u}(I_d)- E \Delta.
\end{align*}
Let $\lchart : \calU_{I_d} \to \R^m$ be the normal coordinates for $\grpG$ about $I_d$ so that $\varepsilon = \lchart(E) = \log(E)^\vee$.
The pre-observer ($\Delta \equiv 0$) dynamics of $\varepsilon$ about $\varepsilon = 0$ are exactly
\begin{align}
    \dot{\varepsilon}
    &=  \mr{A}_t \varepsilon,
    \notag \\
    \mr{A}_t \varepsilon
    &= \left( \tD_E |_{I_n} \Lambda(E, u)[\varepsilon^\wedge] \right)^\vee. \notag
\end{align}

Let the output space $\calN$ be the Euclidean space $\R^n$, and let the configuration output $h: \calG \to \R^n$ be any map; that is, $h$ is not necessarily equivariant.
Then the EqF (\ref{eq:eqf_group_observer}-\ref{eq:eqf_riccati}) with standard output matrix specialises to the invariant extended Kalman filter (IEKF) proposed in \cite{2017_Barrau_tac}.

In \cite{2017_Barrau_tac}, the pre-observer dynamics of $\varepsilon$ are shown to be exactly linear as a consequence of the novel `group affine' property.
In this case the IEKF is locally asymptotically stable with a constant convergence radius \cite{2017_Barrau_tac}.


\bibliographystyle{plain}
\bibliography{references}



\newcommand{\biopicwidth}{1in}

\end{document}